\documentclass[letterpaper, 10pt, conference]{ieeeconf}  

\IEEEoverridecommandlockouts                             

\usepackage{tikz}
\usetikzlibrary{arrows.meta,positioning}
\usepackage{tabularx}
\usepackage{siunitx}
\usepackage{xcolor}

\usepackage{mymath}
\newtheorem{proposition}{Proposition}
\newtheorem{lemma}{Lemma}
\newtheorem{theorem}{Theorem}
\newtheorem{definition}{Definition}
\usepackage{mathrsfs}
\usepackage{parskip}
\usepackage{algorithm}
\usepackage{algpseudocode}
\usepackage{balance}
\usepackage{subcaption}
\usepackage{graphicx}
\usepackage{pgfplots}
\usepackage{siunitx}
\pgfplotsset{compat=1.18}
\makeatletter
\renewcommand*\env@matrix[1][\arraystretch]{%
  \edef\arraystretch{#1}%
  \hskip -\arraycolsep
  \let\@ifnextchar\new@ifnextchar
  \array{*\c@MaxMatrixCols c}}
\makeatother
\usepackage{aligned-overset}
\usepackage{nicematrix}
\usepackage{booktabs}
\usepackage{multirow}
\usepackage{tikz}
\usepackage{graphicx}
\usepackage{varwidth, tikz, nth}
\usepackage{pgfplots}
\usepackage{pgfplotstable}
\usepgfplotslibrary{fillbetween}
\usepgfplotslibrary{statistics}
\pgfplotsset{compat=1.18}
\usepgfplotslibrary{patchplots,colormaps,colorbrewer,fillbetween}
\usetikzlibrary{shapes, arrows, shapes.misc, arrows.meta, positioning, matrix, calc, fit, fadings, patterns, plotmarks, shapes.geometric, pgfplots.colorbrewer, decorations.markings, decorations.pathmorphing, backgrounds, intersections, decorations.text, decorations.markings}
\usepackage{tikzviolinplots}
\definecolor{legreddraw}{HTML}{E41A1C}
\definecolor{legbluedraw}{HTML}{377EB8}
\definecolor{leggreendraw}{HTML}{4DAF4A}
\definecolor{legvioletdraw}{HTML}{984EA3}
\definecolor{legredfill}{HTML}{B61516}
\definecolor{legbluefill}{HTML}{2C6593}
\definecolor{leggreenfill}{HTML}{3E8C3B}
\definecolor{legvioletfill}{HTML}{7A3E82}

\makeatletter
\def\@citex[#1]#2{\leavevmode
\let\@citea\@empty
\@cite{\@for\@citeb:=#2\do
{\@citea\def\@citea{,\penalty\@m\ }%
\edef\@citeb{\expandafter\@firstofone\@citeb\@empty}%
\if@filesw\immediate\write\@auxout{\string\citation{\@citeb}}\fi
\@ifundefined{b@\@citeb}{\hbox{\reset@font\bfseries ?}%
\G@refundefinedtrue
\@latex@warning
{Citation `\@citeb' on page \thepage \space undefined}}%
{\@cite@ofmt{\csname b@\@citeb\endcsname}}}}{#1}}
\makeatother

\newtheorem{corollary}{Corollary}
\newtheorem{remark}{Remark}

\newcommand{\red}[1]{\textcolor{black}{#1}}

\title{\bf \vspace{0.8em}
Data-Based Clustering and Control of Similar Biological Systems
}

\author{ Peilin Zhang,  Antonis Papachristodoulou, Idris Kempf
\thanks{This work was supported by the Engineering and Physical Sciences Research Council (EPSRC) under the EEBio Programme Grant, EP/Y014073/1. For the purpose of Open Access, the author has applied a CC BY public copyright licence to any Author Accepted Manuscript (AAM) version arising from this submission. }
\thanks{All Authors are with the University of Oxford, Parks Road, Oxford OX1 3PJ. Emails: {\tt \small peilin.zhang@keble.ox.ac.uk}, {\tt \small antonis@eng.ox.ac.uk}, {\tt\small idris.kempf@eng.ox.ac.uk}.}%
}

\begin{document}
\bstctlcite{BSTcontrol}

\maketitle
\thispagestyle{empty}
\pagestyle{empty}

\begin{abstract}
Cybergenetic control of gene expression enables applications in synthetic biology, drug development, and biomanufacturing. Microfluidic platforms allow the parallel control of large cell populations. However, the resulting computational burden and intrinsic biological heterogeneity limit the scalability of conventional control strategies. In this work, we propose a similarity-based framework to reduce the computational requirement of controlling large numbers of dynamical systems. Building on existing data-driven methods for quantifying control-relevant similarity from input–output data, we cluster systems with similar dynamics without requiring explicit system identification. Based on this grouping, we develop a hierarchical leader–follower control architecture, where a single controller is designed for each cluster and applied to all members. This significantly reduces the number of control problems that need to be solved. Furthermore, we analyse the closed-loop behaviour within clusters and develop data-driven conditions under which the clustered closed-loop systems remain well-posed. The proposed approach is demonstrated in simulations of gene expression dynamics, showing that similarity-based grouping enables scalable and reliable control of heterogeneous biological systems.
\end{abstract}


\section{INTRODUCTION}

Cybergenetic technologies that apply the principles of control theory and communication (cybernetics) to genetics and living biological systems have enabled controlled regulation of gene expression, growth, and metabolic activity in living cells. Applications include personalised therapy, scalable bioproduction, and the design of robust synthetic circuits~\cite{KHAMMASHREVIEW}. However, biological systems exhibit significant cell-to-cell variability even between genetically identical cells, challenging traditional model-based control. Instead, researchers increasingly seek to control many biological systems individually and in parallel. Controlled variables typically include gene expression levels, measured via fluorescent reporters, and growth rate or metabolite production, actuated through optogenetic, chemical, or mechanical inputs. 

Recent advances in microfluidic platforms and control algorithms have enabled the precise regulation of individual cells in parallel. While these methods rely on either model-based controllers~\cite{automated_optogenetic_control_milias_argeitis_206,shaping_bacterial_populationChait2017} or AI-based methods~\cite{deep_mpc_thousands_of_single_cells_Lugagne2024,10178155}, which are labour-intensive to train and must be re-trained for each application, data-driven approaches such as data-enabled predictive control (DeePC)~\cite{DEEPC} provide a framework that avoids explicit system identification or training. In our recent work~\cite{perreault2026}, we developed a data-based predictive controller for a two-input two-output genetic circuit. By accommodating input nonlinearities via variable transformations, the method outperformed both model-driven and data-driven baselines in terms of performance and sample efficiency.

Despite advances, scalability remains fundamentally limited because of computational requirements. Existing approaches typically compute control inputs independently for each cell, resulting in a computational cost that grows with the population size. In large-scale settings, such as in microfluidic platforms, the number of controlled systems can become a limiting factor, as thousands of optimisation or control problems must be solved in parallel in real time. To enable scalable control of many biological systems, a natural direction is to reuse or combine information across systems rather than optimise for each one individually. 

Clustering and grouping strategies have been widely studied for dynamical systems. In order to tackle the scarcity of physical data, existing approaches group systems based on similarity of observed trajectories and learned representations and have been applied to enable shared data in cases such as transfer learning~\cite{wang2024learningcontrolsimilarityheterogeneous}~\cite{Li_2024}, or federated learning~\cite{vankan2025federateddeepcborrowingdata}. Similarity-based clustering has also found applications in the control of nonlinear systems~\cite{HAN2024103252}. However, these approaches are either not used for clustered control directly, or rely on similarity measures that are not explicitly tied to control objectives and closed-loop behaviour. Clustered control was proposed in~\cite{MUNSER20204623}~\cite{labella2020supervisedmpccontrollargescale} for example, but it relied on the knowledge of explicit system models. However, model-based clustering is impractical in large-scale settings, where reliable model identification is difficult or infeasible. This limitation motivates the need for data-driven clustering methods that can guide grouping decisions for direct sharing of controllers while ensuring some closed-loop guarantees.

In this work, we address the problem of scalable control for large clusters of heterogeneous dynamical systems by leveraging control-relevant and data-driven similarity, rather than heuristic similarity measures. Building on existing data-driven formulations of the gap metric~\cite{padoan2022}, we exploit it for similarity-based clustering. We propose a hierarchical leader--follower control architecture, in which computationally intensive operations are performed only by the cluster leaders, thereby significantly reducing the computational requirements. Furthermore, we establish well-posedness and derive a closed-loop performance bound for the resulting similarity-aware control framework. Finally, we combine similarity-aware control framework with data-based linear quadratic control to build a pipeline and apply it to control biological systems in simulation.

The paper is organised as follows. We introduce the background theory in Section~\ref{sec:background}. Then we develop theory on well-posedness \red{in Section~\ref{sec:theory}} in order to theoretically validate the method. In Section~\ref{sec:sim}, simulation experiments are conducted and results are shown.


\paragraph*{Notation and Definitions}

Let $\mathbb{Z}_{>0}$ denote the set of strictly positive integers and $\R$ the set of real numbers. The Euclidean 2-norm and induced operator norm are denoted by $\|\cdot\|$, and the Kronecker product by $\otimes$. For a linear operator $H:\mathcal{X}\rightarrow\mathcal{Y}$, its image and kernel are denoted by $\mathrm{im}(H)\eqdef\{Hx:x\in\mathcal{X}\}$ and $\ker(H)\eqdef\{x\in\mathcal{X}:Hx=0\}$, respectively, and its adjoint by $H^*:\mathcal{Y}\rightarrow\mathcal{X}$. For a closed subspace $V\subseteq\mathcal H$ of a Hilbert space $\mathcal H$, $P_V:\mathcal H\rightarrow V$ denotes the orthogonal projector onto $V$, defined by $P_Vx\eqdef\arg\min_{v\in V}\|x-v\|$, and let $P_V^\perp\eqdef I-P_V=P_{V^\perp}$. If $Q_V$ has orthonormal columns spanning $V$, then $P_V=Q_VQ_V^\Tr$~\cite{matrix_computation}. For two subspaces $V_1,V_2\subseteq V$, we write $V=V_1\oplus V_2$ if $V=V_1+V_2$ and $V_1\cap V_2=\{0\}$. For a signal $x=\col(x(1),\ldots,x(T))$, $H_L(x)$ denotes the block Hankel matrix of depth $L$,
\begin{equation}\label{eq:hankel}
H_L(x)=
{\setlength{\arraycolsep}{3pt}
\begin{pmatrix}
x(1) & x(2) & \dots & x(T\!-\!L\!+\!1) \\[-.4em]
\vdots & \vdots & \ddots & \vdots \\
x(L) & x(L\!+\!1) & \dots & x(T)
\end{pmatrix}}.
\end{equation}
An $m$-dimensional signal $x$ is persistently exciting of order $L$ if $\rank H_L(x)=mL$.

\section{Background and Preliminaries}\label{sec:background}

The data-driven representation used in this paper relies on Willems'
Fundamental Lemma~\cite{Willems2004PE} and is therefore restricted to
linear time-invariant (LTI) systems. We use complementary operator and
behavioural descriptions of such systems. Let $\mathcal U$ and
$\mathcal Y$ be Hilbert spaces of input and output signals, and let the
plant and controller be linear operators
$G:D(G)\subseteq\mathcal U\rightarrow\mathcal Y$ and
$C:D(C)\subseteq\mathcal Y\rightarrow\mathcal U$, respectively, where
$D(\cdot)$ denotes the domain of an operator. For a finite horizon $L$,
let $\mathcal U_L=\R^{mL}$, $\mathcal Y_L=\R^{pL}$, and
$\mathcal H_L=\mathcal U_L\oplus\mathcal Y_L =\R^{(m+p)L}$, and denote the
corresponding finite-horizon operators by
$G_L:\mathcal U_L\rightarrow\mathcal Y_L$ and
$C_L:\mathcal Y_L\rightarrow\mathcal U_L$. Their graph and inverse
graph are
$\mathcal G_L(G)\eqdef\{\mathrm{col}(u,G_Lu):u\in\mathcal U_L\}$ and
$\mathcal G_L^{-1}(C)\eqdef
\{\mathrm{col}(C_Ly,y):y\in\mathcal Y_L\}$, respectively.
The subscript $L$ is omitted whenever the horizon is clear from context.

In the behavioural description, for an LTI realisation $(A_G,B_G,C_G,D_G)$ of order $n$, let $\mathcal B_G\eqdef\{(u,y)\in\mathcal U\times\mathcal Y:\exists\,x\in(\R^n)^{\mathbb Z_{\geq0}}\ \mathrm{s.t.}\ x_{t+1}=A_Gx_t+B_Gu_t,\ y_t=C_Gx_t+D_Gu_t,\ \forall t\in\mathbb Z_{\geq0}\}$ denote the set of all admissible input--output trajectories of $G$.

\begin{definition}[Restricted Behaviour~\cite{padoan2022}]
For $L\in\mathbb Z_{>0}$, the restricted behaviour
$\mathcal B_G|_L$ is the set of all length-$L$ restrictions of
trajectories in $\mathcal B_G$. For an LTI realisation of order $n$,
$\mathcal B_G|_L\eqdef\{\mathrm{col}(u,G_Lu+O_Lx_0):
u\in\R^{mL},\,x_0\in\R^n\}$, where $O_L$ is the finite-horizon
observability matrix.
\end{definition}

The restricted behaviour $\mathcal B_G|_L\subseteq\mathcal H_L$ provides a finite-dimensional representation of the system behaviour. The subset of $\mathcal B_G|_L$ corresponding to $x_0=0$ is the finite-horizon graph $\mathcal G_L(G)$. Under the conditions of Willems' Fundamental Lemma, $\mathcal B_G|_L$ can be obtained directly from a sufficiently informative input--output trajectory~\cite{Willems2004PE,padoan2022}. In particular, for a controllable LTI system of order $n$, if the input $u$ is persistently exciting of order $L+n$, then \(\mathcal{B}_G|_L=\mathrm{im}\, H_L^{\mathrm{io}} \eqdef \mathrm{im}{\renewcommand{\arraystretch}{0.8} \begin{bmatrix}H_L(u)\\ H_L(y)\end{bmatrix}}.\)We next consider well-posedness of the feedback interconnection shown in Fig.~\ref{fig:system}.

\begin{definition}[Well-posedness~\cite{Feintuch1998}]
\label{def:well-pose}
For the feedback interconnection in Fig.~\ref{fig:system}, define
$S\eqdef{\setlength{\arraycolsep}{2pt}\begin{bmatrix}I&C\\[-.2em]G&-I\end{bmatrix}}$. The interconnection $[G,C]$ is said to be \emph{well-posed} if $S:D(G)\oplus D(C)\rightarrow\mathcal U\oplus\mathcal Y$ is bijective.
\end{definition}

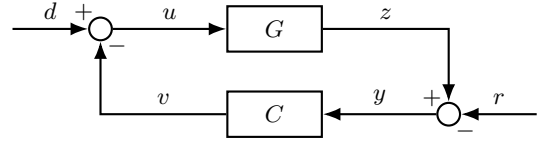
\begin{figure}[t]
\centering
\begin{tikzpicture}[
    >=Latex,
    block/.style={draw, thick, minimum width=1.25cm, minimum height=0.6cm},
    sum/.style={draw, circle, thick, minimum size=3mm, inner sep=0pt},
    every node/.style={font=\small}
]
\node[sum] (s1) at (0,0) {};
\node[block, right=1.5cm of s1] (G) {$G$};
\node[block, below=0.5cm of G] (C) {$C$};
\node[sum, right=1.5cm of C] (s2) {};
\node at ($(s1)+(0.22,-0.22)$) {$-$};
\node at ($(s1)+(-0.22,0.22)$) {$+$};
\node at ($(s2)+(-0.22,0.22)$) {$+$};
\node at ($(s2)+(0.22,-0.22)$) {$-$};
\draw[->, thick] ([xshift=-1cm]s1.west) -- (s1.west) node[midway, above] {$d$};
\draw[->, thick] (s1) -- (G) node[midway, above] {$u$};
\draw[->, thick] (G) -- (G.center -| s2.center) node[midway, above] {$z$} -- (s2.north);
\draw[->, thick] (s2.west) -- (C.east) node[midway, above] {$y$};
\draw[->, thick] (C.west) -- (C.center -| s1.center) node[midway, above] {$v$} -- (s1.south);
\draw[->, thick] ([xshift=1cm]s2.east) -- (s2.east) node[midway, above] {$r$};
\end{tikzpicture}
\caption{Standard feedback configuration.}\label{fig:system}
\end{figure}

Hence, every admissible pair of exogenous signals determines a unique
pair of internal signals. On a finite horizon,
Def.~\ref{def:well-pose} is applied directly to $G_L$ and $C_L$,
i.e., by requiring $S_L\eqdef{\setlength{\arraycolsep}{2pt}\begin{bmatrix}I&C_L\\[-.1em]G_L&-I\end{bmatrix}}$ to be bijective.

To compare systems directly from data, we use the $L$-gap, a
finite-horizon counterpart of the gap metric~\cite{Feintuch1998,padoan2022}.

\begin{definition}[$L$-gap~\cite{padoan2022}]\label{def:gap}
Let $\mathcal B_1|_L$ and $\mathcal B_2|_L$ be two restricted
behaviours, with orthogonal projectors $P_1$ and $P_2$. Their $L$-gap
is defined as
$\mathrm{gap}_L(\mathcal B_1|_L,\mathcal B_2|_L)
\eqdef\|P_1-P_2\|
=\max\{\|P_2^\perp P_1\|,\|P_1^\perp P_2\|\}$. If $\dim B_1|_L\equiv\dim B_2|_L$, $\|P_2^\perp P_1\|=\|P_1^\perp P_2\|$. 
\end{definition}

For restricted behaviours of equal dimension, the $L$-gap is
equivalently the sine of their largest principal angle and can
therefore be computed directly from bases of the corresponding
subspaces~\cite{padoan2022}. This provides the similarity measure used
for clustering in this paper.

\section{Main Results} \label{sec:theory}

\subsection{Clustering}

We first derive finite-horizon conditions under which a controller designed for one system can be applied to another while preserving well-posedness and bounding performance degradation.

\begin{lemma}[Injectivity]\label{thm:injectivity}
The finite-horizon closed-loop operator $S_L$ is injective iff $\mathcal G_L(G)\cap\mathcal G_L^{-1}(-C)=\{0\}.$
\end{lemma}

\begin{proof}
$S_L$ is injective iff $\ker(S_L)=\{0\}$. For $w=\mathrm{col}(u,y)$, $w\in\ker(S_L)$ iff $u=-C_Ly$ and $y=G_Lu$, which is equivalent to $w\in\mathcal G_L(G)\cap\mathcal G_L^{-1}(-C)$. Hence, $\ker(S_L)=\mathcal G_L(G)\cap\mathcal G_L^{-1}(-C)$, and the result follows.
\end{proof}

\begin{lemma}[Surjectivity]\label{thm:surjectivity}
The finite-horizon closed-loop operator $S_L$ is surjective iff $\mathcal G_L(G)+\mathcal G_L^{-1}(-C)=\mathcal H_L.$
\end{lemma}

\begin{proof}
$S_L$ is surjective iff $\mathrm{im}(S_L)=\mathcal H_L$. For $w=\mathrm{col}(u,y)$, $S_Lw = \mathrm{col}(u,G_Lu)+\mathrm{col}(C_Ly,-y)$, where $\mathrm{col}(u,G_Lu)\in\mathcal G_L(G)$ and $\mathrm{col}(C_Ly,-y)\in\mathcal G_L^{-1}(-C)$. Since $u$ and $y$ are arbitrary, $\mathrm{im}(S_L)=\mathcal G_L(G)+\mathcal G_L^{-1}(-C)$, and the result follows.
\end{proof}

\begin{corollary}[Bijectivity]\label{thm:bijectivity}
The finite-horizon closed-loop operator $S_L$ is bijective iff $\mathcal H_L=\mathcal G_L(G)\oplus\mathcal G_L^{-1}(-C)$.
\end{corollary}

The well-posedness characterisation is expressed in terms of the finite-horizon graph $\mathcal G_L(G)$, rather than the full restricted behaviour $\mathcal B_G|_L = \mathcal G_L(G) + \{\mathrm{col}(0,O_Lx_0):x_0\in\R^n\}$. We therefore introduce the corresponding graph-based finite-horizon distance.

\begin{definition}[Graph $L$-gap]\label{def:graphgap}
Let $G_1,G_2:\mathcal U_L\rightarrow\mathcal Y_L$, and let
$P_{\mathcal G_1}$ and $P_{\mathcal G_2}$ denote the orthogonal
projectors onto $\mathcal G_1\eqdef \mathcal G_L(G_1)$ and $\mathcal G_2\eqdef \mathcal G_L(G_2)$,
respectively. We define 
\begin{align*}
\mathrm{gap}^{\mathcal G}_L(G_1,\!G_2)
\!\eqdef\!\|P_{\mathcal G_1}\!\!-\!P_{\mathcal G_2}\|
\!=\! \max\left\{\!
\|P_{\mathcal G_2}^\perp\! P_{\mathcal G_1}\|,\!
\|P_{\mathcal G_1}^\perp\! P_{\mathcal G_2}\|
\!\right\}.
\end{align*}
If $\dim\mathcal G_1\equiv\dim\mathcal G_2$, then $\|P_{\mathcal G_2}^\perp P_{\mathcal G_1}\|=\|P_{\mathcal G_1}^\perp P_{\mathcal G_2}\|$. 
\end{definition}

In contrast to the $L$-gap in Def.~\ref{def:gap}, which compares restricted behaviours including arbitrary initial conditions, $\mathrm{gap}^{\mathcal G}_L$ compares the corresponding zero-state input--output maps. It is related to well-posedness of the closed loop in Thm.~\ref{thm:gaptheorem}.

\begin{theorem}\label{thm:gaptheorem}
Let $C$ be such that $[G_1,C]$ is well-posed on the
finite horizon $L$. Then $[G_2,C]$ is well-posed if
\begin{equation}\label{eq:gapthreshold}
\mathrm{gap}^{\mathcal G}_L(G_1,G_2)<
\varepsilon_L
\eqdef
1-\mathrm{gap}^{\mathcal G}_L(G_1,C^*).
\end{equation}
\end{theorem}
\begin{proof}
By Cor.~\ref{thm:bijectivity},
$[G_i,C]$ is well-posed iff $\mathcal H_L=\mathcal G_L(G_i)\oplus \mathcal G_L^{-1}(-C)$. Since $\dim(\mathcal G_L(G_i))\!=\!mL$, $\dim(\mathcal G_L^{-1}(-C))\!=\!pL$, and $\dim(\mathcal H_L)\!=\!(m+p)L$, this is equivalent to $\mathcal G_L(G_i)^\perp\!\!\cap\!\bigl(\mathcal G_L^{-1}(-C)\bigr)^\perp\!=\!\{0\}$. By Lemma~\ref{app:thm:orthogonal_complement}, $\bigl(\mathcal G_L^{-1}(-C)\bigr)^\perp\!\!=\!\mathcal G_L(C^*)$, so $\mathcal G_L(G_i)^\perp\!\!\cap\!\mathcal G_L(C^*)\!=\!\{0\}$, where $\dim(\mathcal G_L(G_i))\!\equiv\!\dim(\mathcal G_L(C^*))\!=\!mL$. By Lemma~\ref{app:thm:gapintersection}, $\mathcal G_L(G_i)^\perp\cap\mathcal G_L(C^*)=\{0\}$ iff $\mathrm{gap}^{\mathcal G}_L(G_i,C^*)<1$. Since $[G_1,C]$ is well-posed, $\mathrm{gap}^{\mathcal G}_L(G_1,C^*)<1$, and therefore $\varepsilon_L>0$. By the triangle inequality,   
\begin{align*}
\mathrm{gap}^{\mathcal G}_L(G_2,C^*)
&\leq
\mathrm{gap}^{\mathcal G}_L(G_1,G_2)
+\mathrm{gap}^{\mathcal G}_L(G_1,C^*)\\
&<\varepsilon_L
+\mathrm{gap}^{\mathcal G}_L(G_1,C^*)
=1.\\[-3em]
\end{align*}
\end{proof}

Theorem~\ref{thm:gaptheorem} requires the graph $L$-gap to be computed from input--output data. For restricted behaviours,~\cite{padoan2022} computes the $L$-gap by applying an SVD to the stacked input--output Hankel matrices $H_{L,1}^{\mathrm{io}}$ and $H_{L,2}^{\mathrm{io}}$, whose left singular vectors provide an orthonormal basis for the corresponding subspace. These matrices span $\mathcal{B}_1\vert L$ and $\mathcal{B}_2\vert L$, including trajectories associated with arbitrary initial conditions. In contrast, the graph $\mathcal G_L(G)$ contains only zero-state trajectories. To extract this subspace from data, we partition Hankel matrices of depth $T_{\mathrm{ini}}+L$ as
\begin{align}\label{eq:datadefinition}
&H_{T_{\mathrm{ini}}+L}(u)\!=\!\col\left(U_p, U_f\!\right),
&H_{T_{\mathrm{ini}}+L}(y)\!=\!\col\left(Y_p, Y_f\!\right),
\end{align}
where $U_p$ and $Y_p$ have $mT_{\mathrm{ini}}$ and $pT_{\mathrm{ini}}$ rows, respectively.
The following proposition shows that imposing a zero input--output past extracts the finite-horizon graph from~\eqref{eq:datadefinition}.
\begin{proposition}[Data-driven finite-horizon graph]
\label{prop:datagraph}
Consider a controllable and observable LTI system of order $n$, let $T_{\mathrm{ini}}\geq n$, and assume that $u$ is persistently exciting
of order $T_{\mathrm{ini}}+L+n$. Let $N_p$ be a matrix whose columns form a basis for $\ker(\col(U_p, Y_p))$. Then $\mathcal G_L(G) =\mathrm{im}\left(\col(U_f, Y_f)N_p\right)$.
\end{proposition}
\begin{proof}
By Willems' Fundamental Lemma, every length-$(T_{\mathrm{ini}}+L)$
trajectory can be written as
\[
\col(u_p,y_p,u_f,y_f)=
\col\left(U_p, Y_p, U_f, Y_f\right) g,
\]
for some coefficient vector $g$. Let $g=N_p\alpha$. Then $U_pg=0=u_p$ and $Y_pg=0=y_p$. With $u_p=0$, $y_p=O_{T_{\mathrm{ini}}}x(0)$, and since the system is observable and $T_{\mathrm{ini}}\geq n$, $y_p=0$ implies $x(0)=0$. Consequently, $\col(U_f, Y_f) g \in\mathcal G_L(G)$. Conversely, any trajectory in $\mathcal G_L(G)$ starts from the zero state and may therefore be preceded by $T_{\mathrm{ini}}$ zero input--output samples. The resulting length-$(T_{\mathrm{ini}}+L)$ trajectory is represented by the Hankel matrix by the Fundamental Lemma with coefficient vector $g$ satisfying $U_pg=Y_pg=0$, and hence $g\in\mathrm{im}(N_p)$.
\end{proof}

\subsection{Performance}

We next quantify the difference between the finite-horizon closed-loop trajectories of a leader $G_1$ and follower $G_2$ when the same controller $C$ and exogenous input $e\eqdef\col(d,r)$ are used. By Thm.~\ref{thm:gaptheorem}, $\mathrm{gap}^{\mathcal G}_L(G_1,G_2)<1-\mathrm{gap}^{\mathcal G}_L(G_1,C^*)$ guarantees that both interconnections are well-posed.

For each $i\in\{1,2\}$, well-posedness gives a unique decomposition $e=w_i+n_i$, where $w_i\in\mathcal G_L(G_i)$ and $n_i\in\mathcal G_L^{-1}(-C)$. Let $P_C$ denote the orthogonal projector onto $\mathcal G_L(C^*)$. Since $\mathcal G_L^{-1}(-C)=\mathcal G_L(C^*)^\perp$, projecting the decomposition onto $\mathcal G_L(C^*)$ gives
\begin{align}\label{eq:closedlooptrajectory}
P_Ce=P_Cw_i=P_{C|i}w_i
\quad\Longrightarrow\quad
w_i=P_{C|i}^{-1}P_Ce,
\end{align}
where $P_{C|i}:\mathcal G_L(G_i)\rightarrow\mathcal G_L(C^*)$ denotes the restriction of $P_C$ to $\mathcal G_L(G_i)$. Since $\mathrm{gap}^{\mathcal G}_L(G_i,C^*)<1$, Lemma~\ref{app:thm:restrictedprojector} guarantees that $P_{C|i}$ is bijective.

The right-hand side of~\eqref{eq:closedlooptrajectory} allows trajectory differences between leader and followers to be bounded as shown in Prop.~\ref{thm:trajectorybound}.

\begin{proposition}\label{thm:trajectorybound}
Under the assumptions of Thm.~\ref{thm:gaptheorem}, the corresponding closed-loop trajectories satisfy $\|w_2-w_1\|\leq\eta_w\|e\|$, where
\begin{align}\label{eq:trajectorybound}
\eta_w\eqdef
\frac{D_{12}}{\sqrt{1-\delta^2}}
\frac{\sqrt{2-(\delta+D_{12})^2}}
{\sqrt{1-(\delta+D_{12})^2}},
\end{align}
and $D_{12}\eqdef\mathrm{gap}^{\mathcal G}_L(G_1,G_2)$ and $\delta\eqdef\mathrm{gap}^{\mathcal G}_L(G_1,C^*)$.
\end{proposition}

\begin{proof}
From~\eqref{eq:closedlooptrajectory}, $w_2-w_1=(P_{C|2}^{-1}-P_{C|1}^{-1})P_Ce$. For any $x\in\mathcal G_L(C^*)$, let $b_i=P_{C|i}^{-1}x\in\mathcal G_L(G_i)$ and $\tilde b_2=P_2b_1$, where $P_2$ projects onto $\mathcal G_L(G_2)$. Then $\|b_1-\tilde b_2\|=\|(P_1-P_2)b_1\|\leq D_{12}\|P_{C|1}^{-1}\|\|x\|$. Moreover, since $P_Cb_1=P_Cb_2=x$, $b_2-\tilde b_2=P_{C|2}^{-1}P_C(b_1-\tilde b_2)$. Since $b_2-\tilde b_2\in\mathcal G_L(G_2)$ and $b_1-\tilde b_2\in\mathcal G_L(G_2)^\perp$, these terms are orthogonal, and therefore
\begin{align}\label{eq:inverseprojectorbound}
\|(P_{C|2}^{-1}-P_{C|1}^{-1})x\|^2
\leq
\left(\!1\!+\!\|P_{C|2}^{-1}\|^2\!\right)
D_{12}^2\|P_{C|1}^{-1}\|^2\|x\|^2.
\end{align}
Let $\delta_i\eqdef\mathrm{gap}^{\mathcal G}_L(G_i,C^*)$. Lemma~\ref{app:thm:restrictedprojector} gives $\|P_{C|i}^{-1}\|\leq1/\sqrt{1-\delta_i^2}$, while $\delta_1=\delta$ and $\delta_2\leq\delta+D_{12}$ by the triangle inequality. Substitution into~\eqref{eq:inverseprojectorbound}, together with $\|P_C\|=1$, yields~\eqref{eq:trajectorybound}.
\end{proof}

Consider the quadratic finite-horizon performance measure
\begin{align}\label{eq:cost}
J(w)
\eqdef\sum_{k=0}^{L-1}\left(y_k^\Tr Qy_k+u_k^\Tr Ru_k\right)
=w^\Tr Ww,
\end{align}
where $w=\col(u,y)$, $Q=Q^\Tr\succeq0$, $R=R^\Tr\succeq0$, and $W\eqdef\diag(I_L\otimes R,I_L\otimes Q)$. Proposition~\ref{thm:performancebound} formulates a bound on the performance difference $J(w_2)-J(w_1)$ for the case that the same exogenous inputs $e$ are applied to each system.

\begin{proposition}\label{thm:performancebound}
Under the conditions of Prop.~\ref{thm:trajectorybound}, the performance difference satisfies $|J(w_2)-J(w_1)|\leq\eta_J\|e\|^2$, where $\eta_J\eqdef\eta_w\|W\|\left(1/\sqrt{1-\delta^2}+1/\sqrt{1-(\delta+D_{12})^2}\right)$.
\end{proposition}

\begin{proof}
Using~\eqref{eq:cost}, $J(w_2)-J(w_1)=(w_2-w_1)^\Tr W(w_2+w_1)$, so $|J(w_2)-J(w_1)|\leq\|W\|\|w_2-w_1\|\|w_2+w_1\|$. From~\eqref{eq:closedlooptrajectory} and Lemma~\ref{app:thm:restrictedprojector}, $\|w_1+w_2\|\leq\left(1/\sqrt{1-\delta^2}+1/\sqrt{1-(\delta+D_{12})^2}\right)\|e\|$. Combining this with Prop.~\ref{thm:trajectorybound} gives the result.
\end{proof}

\section{Case Study: Cybergenetic Control} \label{sec:sim}

We consider the application of our data-based clustering and control framework to bacteria cultured in microfluidic platforms~\cite{KHAMMASHREVIEW}, which can involve populations of up to $10^6$ cells. Controlling such large populations at experimentally relevant sampling rates can impose a substantial computational burden, motivating approaches that reduce the number of controller synthesis problems while maintaining sufficient control performance. Here, the approach is evaluated using $N=100$ systems on a standard laptop (Intel Core i7 processor, 8 GB RAM).

\subsection{Biological Model}

In many cybergenetic applications, the controlled quantities involve the expression levels of one or several genes of interest, measured indirectly via fluorescent reporter proteins. As an example, we simulate an \textit{E. coli} strain expressing a GFP reporter whose concentration is regulated by the optogenetic \textit{CcaS/CcaR} system~\cite{Olson2014}, with parameter values chosen from ranges reported in~\cite{weisse_mechanistic_links,mRNA_degradation,KLUMPP20091366}. The dynamics of this system can be approximated as
\begin{subequations}\label{eq:systems}
\begin{align}
\dot{m}(t) &= k_{m} + f(u(t)) - \gamma_m\, m(t), \\
\dot{M}(t) &= k_M\, m(t) - \gamma_M\, M(t),
\end{align}
\end{subequations}
where $m(t)$ is the mRNA concentration, $M(t)$ is the GFP concentration, $k_m=\SI{1.3}{\per\minute}$ is the basal transcription rate, $k_M=\SI{1}{\per\minute}$ is the translation rate, and $\gamma_m=\SI{0.14}{\per\minute}$ and $\gamma_M=\SI{0.02}{\per\minute}$ are the mRNA and protein degradation/dilution rates, respectively. The optogenetic actuation $f(u(t))$ is modelled by a Hill function, $f(u)=\alpha u^n/(K_f^n+u^n)$, with $\alpha=\SI{1.3}{\per\minute}$, $K_f=1$, and $n=2$, where $u(t)\in[0,u_{\max}]$ is the input light intensity. For the following analysis, the system is linearised at $u^\star=K_f$, $m^\star=(k_m+f(u^\star))/\gamma_m$, and $M^\star=k_Mm^\star/\gamma_M$. As the subsequent sections use state feedback, it is assumed that both $m(t)$ and $M(t)$ are measurable, resulting in the state-space representation $A={\setlength{\arraycolsep}{1pt}\begin{bmatrix}-\gamma_m & 0 \\[-.1em] k_M& -\gamma_M\end{bmatrix}}$, $B={\setlength{\arraycolsep}{2pt}\begin{bmatrix}\alpha n / 4K_f\\[-.1em] 0\end{bmatrix}}$, $C=I$, and $D=0$. The system is then discretised using zero-order hold with a sample time of $T_s=\SI{10}{\minute}$.

\subsection{Data-generation}

To capture cell-to-cell variability, $\gamma_m$, $\gamma_M$, $k_m$, and $k_M$ are sampled independently and uniformly between $0.8$ and $1.2$ times their nominal values, consistent with observed stochasticity in gene expression~\cite{stochastic_cells}, while $\alpha$, $K_f$, and $n$ are kept fixed. The current formulation assumes linear dynamics around a nominal operating point. For nonlinear systems, local linearisation may introduce input constraints, which can be accommodated via appropriate input transformations~\cite{perreault2026}.

In practice, the system models are unknown, and the graph subspaces are constructed directly from input--output data using Prop.~\ref{prop:datagraph}. For each system, noise-free input--output data are used to form the partitioned Hankel matrices in~\eqref{eq:datadefinition}, with $L=10$ and $T_{\mathrm{ini}}=n=2$, and the input is chosen persistently exciting of order $T_{\mathrm{ini}}+L+n$. If $N_p$ spans $\ker(\col(U_p,Y_p))$, then $Z\eqdef\col(U_f,Y_f)N_p$ spans $\mathcal G_L(G)$. An SVD of $Z$ provides an orthonormal basis $Q$ for this subspace~\cite{padoan2022}, from which the corresponding projector $P_{\mathcal G}=QQ^\Tr$ and graph $L$-gaps are computed. 

\subsection{Similarity-Based Clustering}

Given the pairwise graph $L$-gap distances, define the distance matrix $D\in\R^{N\times N}$ by $D_{ij}\eqdef\mathrm{gap}^{\mathcal G}_L(G_i,G_j)$. We partition the $N$ systems into $K\leq N$ nonempty, disjoint clusters $\{\mathcal C_1,\ldots,\mathcal C_K\}$ such that $\bigcup_{k=1}^K\mathcal C_k=\{1,\ldots,N\}$. For each cluster $\mathcal C_k$, one system $i_k\in\mathcal C_k$ is designated as the \emph{leader}, while the remaining systems are \emph{followers}, and a single controller $C_k$ designed for the leader is applied to all systems in $\mathcal C_k$. The clustering objective is to group systems with small pairwise graph $L$-gap. In particular, Thm.~\ref{thm:gaptheorem} provides a sufficient leader--follower distance for preserving well-posedness, while Prop.~\ref{thm:performancebound} shows that smaller distances also reduce the bound on performance degradation.

To achieve this, we adopt a farthest-point first clustering strategy. The algorithm selects initial members that are maximally separated in terms of the $L$-gap, thereby ensuring coverage of the diversity in system behaviours and dissimilar systems will not be clustered together. Classical clustering methods, such as $K$-means, may not work well as the metric is not a Euclidean distance.

\begin{figure}[!t]
\centering

\begin{subfigure}{0.96\columnwidth}
    \centering
    \includegraphics[
        width=\linewidth,
        height=0.18\textheight,
        keepaspectratio
    ]{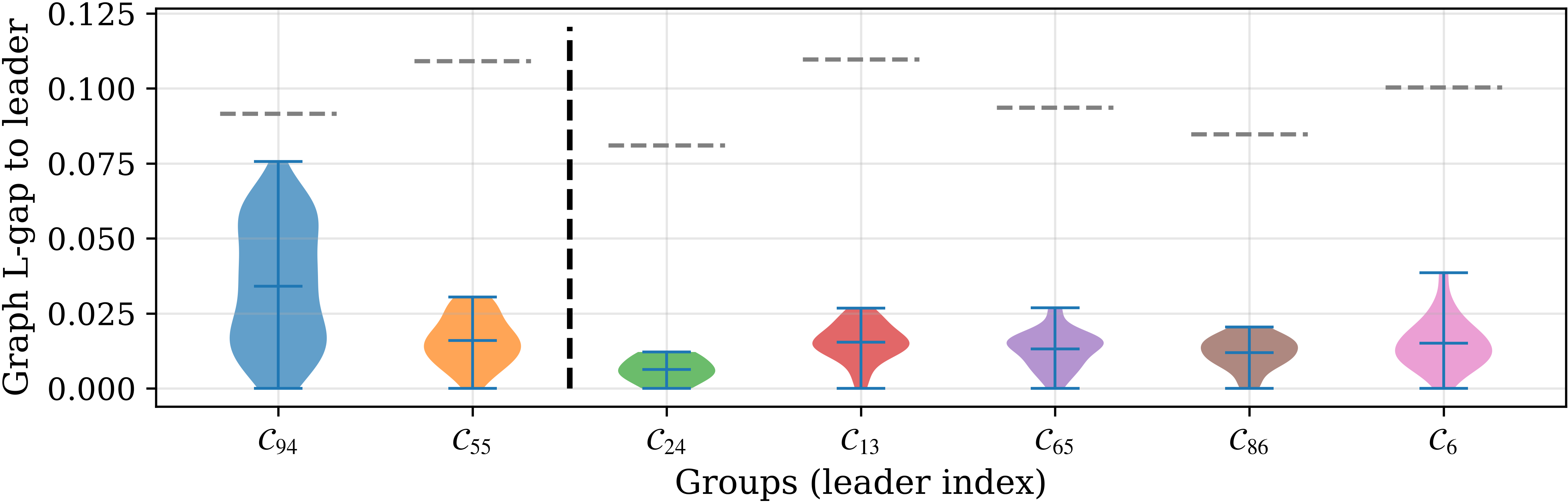}
    \caption{}
    \label{fig:clustering_violin}
\end{subfigure}

\vspace{1em}

\begin{subfigure}{0.96\columnwidth}
    \centering
    \includegraphics[
        width=\linewidth,
        height=0.18\textheight,
        keepaspectratio
    ]{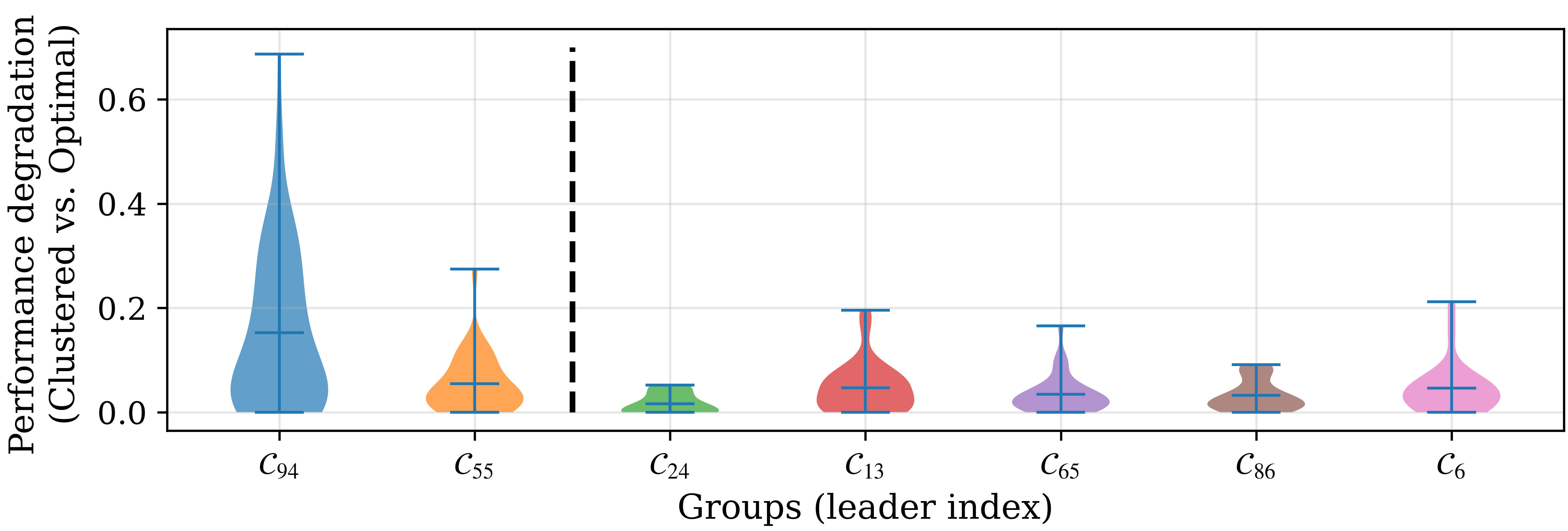}
    \caption{}
    \label{fig:error_violin}
\end{subfigure}

\vspace{1em}

\begin{subfigure}{0.96\columnwidth}
    \centering
    \includegraphics[
        width=\linewidth,
        height=0.18\textheight,
        keepaspectratio
    ]{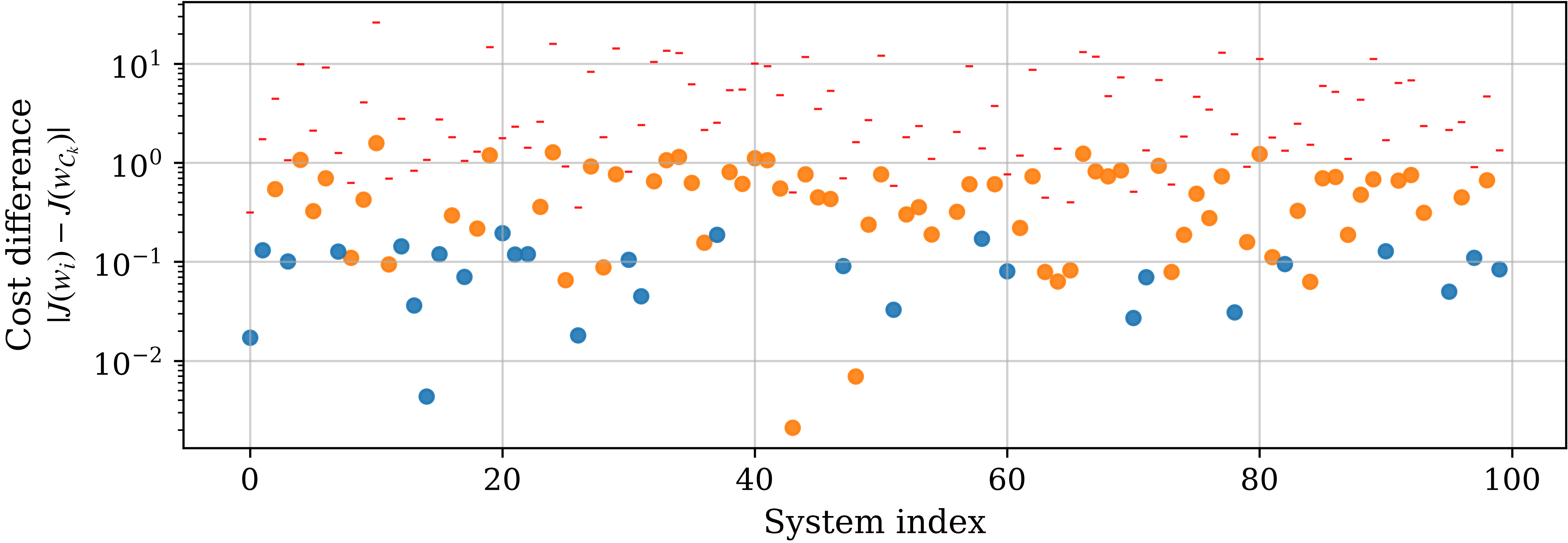}
    \caption{}
    \label{fig:performance_difference}
\end{subfigure}

\vspace{1em}

\begin{subfigure}{0.96\columnwidth}
    \centering
    \includegraphics[
        width=\linewidth,
        height=0.18\textheight,
        keepaspectratio
    ]{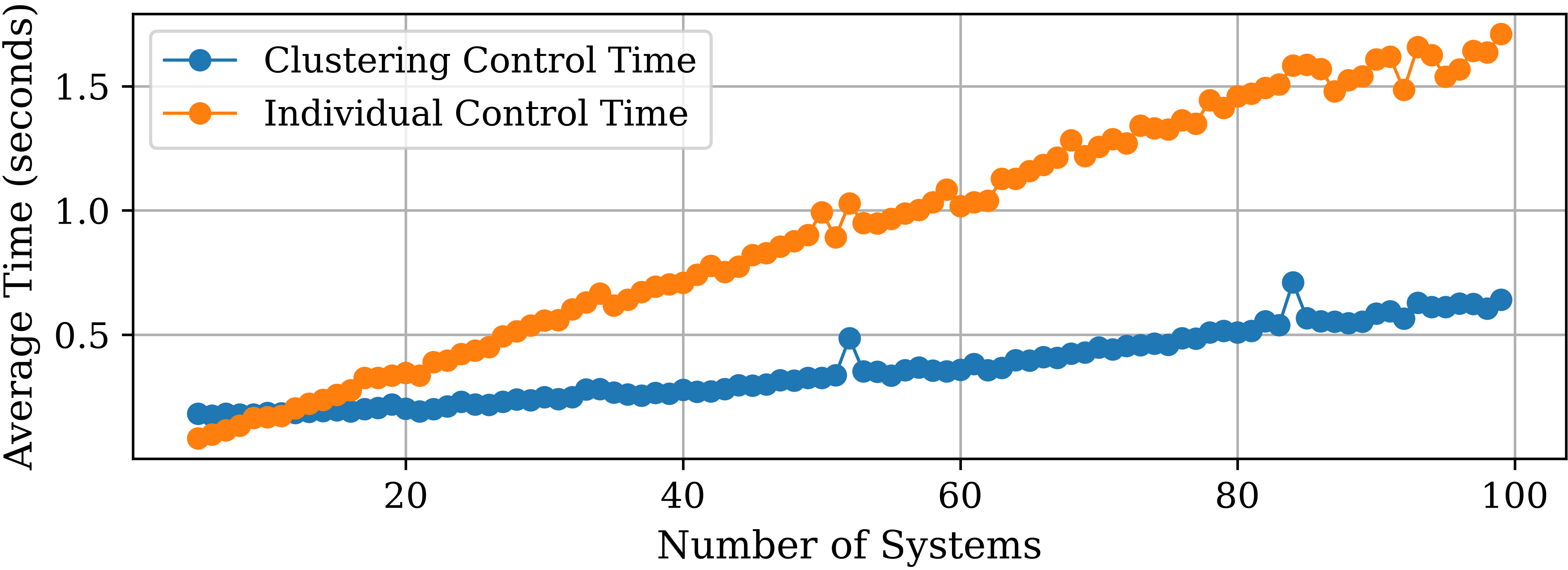}
    \caption{}
    \label{fig:runtime_comparison}
\end{subfigure}

\vspace{0em}

\caption{%
Evaluation of the proposed clustering-based control framework.
(a) Distributions of the graph $L$-gaps between each follower and its assigned leader for clusterings with $K=2$ and $K=5$. Horizontal dashed lines indicate the corresponding $\epsilon_L$ thresholds from Thm.~\ref{thm:gaptheorem}, while the vertical dashed line separates the two clustering configurations.
(b) Performance degradation relative to individually optimised control, measured by $J(w_i)-J(w_i^\star)$.
(c) Absolute cost difference between each follower and its leader for $K=2$, measured by $|J(w_i)-J(w_{\mathcal{C}_k})|$; horizontal red markers indicate the bound from Prop.~\ref{thm:performancebound}.
(d) Runtime comparison between the proposed cluster-and-control method and the baseline in which an individual controller is designed for each system.
}
\label{fig:combined_results}
\end{figure}

The algorithm is applied to cluster the $100$ bacterial systems and the result is shown in Fig.~\ref{fig:clustering_violin}. The y-axis shows the $L$-gap to the leader of each cluster $\mathcal{C}_k$ and the dotted grey lines indicate the well-posedness bound $\epsilon_L$ from Thm.~\ref{thm:gaptheorem}. If only well-posedness is ensured, the systems are clustered into two clusters as shown on the left of the black line. For comparison, the result obtained by enforcing five clusters—--yielding groups with more closely matched behaviours—--is shown on the right of the same line.

\subsection{Leader--Follower Control Architecture}

Based on the clustering obtained in the previous part, we implement a leader--follower control architecture. For each cluster $\mathcal{C}_k$ a controller $C_k$ is designed solely for a representative system (leader).

Here, we employ a data-driven linear quadratic regulator (LQR) formulation~\cite{PersisLQR}. As a preliminary study, we assume that the measurements are noise-free and that all states are directly measurable---an assumption that could be relaxed by using a data-based observer formulation. Here, we use data from the leader only, but future work could exploit clustering to average data across each cluster. Applying the same controller $C_k$ to all systems within a cluster allows multiple systems to share a single control design, thereby reducing the total number of control synthesis problems from $N$ to $K$. For this data-driven LQR implementation, the gain matrix is obtained by solving a convex SDP~\cite[Thm.~4]{PersisLQR}, with decision variables $Q \in \mathbb{R}^{T \times n}$ and $X \in \mathbb{S}_+^{m}$ (here, $n\equiv m$), and two LMIs of sizes $(m+n)$ and $2n$. The dominant computational cost is therefore the SDP solve, which scales as $\mathcal{O}(T^3 n^3)$ to leading order. Consequently, synthesising controllers for $N$ systems scales as $\mathcal{O}(N T^3 n^3)$. By grouping the systems into $K \ll N$ representative clusters and solving one controller per group, this cost reduces to $\mathcal{O}(K T^3 n^3)$. Moreover, while in this preliminary study the gain matrices are computed only once, in practice the input--output data might be updated as the linearisation point of the nonlinear system changes (see e.g.\ our previous work~\cite{perreault2026}). Moreover, for data-based predictive control (DeePC), clustering would allow one to re-use factorisations of the Hessian.

Since the pairwise gaps between each follower and the cluster leader satisfy the assumptions of Thm.~\ref{thm:gaptheorem}, the closed-loop systems are well-posed. Fig.~\ref{fig:error_violin} shows the performance degradation, quantified as the difference in integrated squared errors between clustered control and individually designed control with a disturbance impulse $d_0=2^\Tr$. For all the systems, the error does not change drastically, showing the reliability of the clustering method. In addition, as expected, partitioning the systems into smaller clusters further reduces the degradation, allowing one to trade cluster size for better performance.

Fig.~\ref{fig:performance_difference} shows the cost differences (as defined in~\eqref{eq:cost}) between the leader and follower systems for the two-cluster case and with a disturbance step ($d=1^\Tr$ in Fig.~\ref{fig:system}). The cluster members are re-indexed and separated by the vertical dotted line. All observed differences lie well within the theoretical bounds, confirming the validity of Prop.~\ref{thm:performancebound}. The results show that by partitioning the 100 systems into two clusters, the number of required controllers is reduced from 100 to 2, corresponding to a 98\% reduction in control synthesis effort. Despite this substantial reduction, performance remains consistent across each cluster. When taking the runtime of clustering procedure into consideration, the runtime of this method is still shorter compared to the baseline method for up to at least 100 systems as shown in Fig.~\ref{fig:runtime_comparison}. Overall, the results demonstrate that the proposed method achieves reliable performance with strong theoretical guarantees. Combined with the data-driven formulation, which accommodates system heterogeneity without explicit modelling, the clustering framework enables scalable, parallel control scenarios where controlling many systems individually would otherwise be computationally prohibitive.

\section{Conclusion}

This paper investigated the problem of clustering dynamical systems from a control perspective, with the aim of enabling controller reuse across similar systems. A data-driven framework based on restricted graphs was employed, together with a sufficient, data-driven condition for well-posedness of the clustered closed-loop systems proposed. The framework was evaluated with data-driven LQR control. The results demonstrated that a large number of systems, in this simplified setting, can efficiently be clustered into a small number of groups while incurring only a small degradation in closed-loop performance. In this preliminary case study, all system states were assumed to be directly measurable and the measurements were noise-free. In practice, data-driven LQR could be combined with a shared, data-driven observer design, which would introduce additional considerations in both clustering and performance guarantees.

The method reduces online computational complexity by clustering $N$ systems into $K\ll N$ groups, such that controller design and implementation can be performed at the cluster level. The computational savings are limited by the overhead of the clustering procedure as pair-wise calculation scales poorly with number of systems . However, for more demanding control algorithms, such as data-based predictive control, the benefits of clustering are expected to significantly outweigh this overhead. This is particularly relevant for nonlinear systems, where input--output data may need to be updated as the operating point changes, requiring repeated controller synthesis.

This observation raises several directions for future work. First, when input--output data is updated online, the corresponding $L$-gap between systems may become time-varying, suggesting the need for adaptive gap estimation and clustering strategies. Second, the current framework relies on computing all pairwise gaps, which scales quadratically with the number of systems. Updating the input--output data also requires re-factorising the Hankel matrices to extract a subspace basis. Exploiting structural properties such as the triangle inequality of the gap may enable more efficient clustering schemes, for example by reducing the number of required gap evaluations or enabling hierarchical clustering strategies.

\newpage

\bibliographystyle{IEEEtran}
\small
\bibliography{mybib_clean_abbrev}

\balance

\newpage

\appendix

\renewcommand{\thetable}{A.\arabic{table}}
\setcounter{table}{0}

\renewcommand{\thelemma}{A.\arabic{lemma}}
\setcounter{lemma}{0}

\renewcommand{\theremark}{A.\arabic{remark}}
\setcounter{remark}{0}

\begin{lemma}[Orthogonal Complement]\label{app:thm:orthogonal_complement}
For a bounded linear operator $C : H^2 \to H^2$, define its graph as $\mathcal G(C) = \left\{ \mathrm{col}(u, Cu) : u \in H^2 \right\}$. Then, the orthogonal complement of the graph is given by $\mathcal G(C)^\perp = \left\{ \mathrm{col}(-C^* v, v) : v \in H^2 \right\}$.
\end{lemma}
\begin{proof}
For $H^2 \oplus H^2$ with the standard inner product
\[
\left\langle 
\mathrm{col}(f_1, g_1), 
\mathrm{col}(f_2, g_2)
\right\rangle
= \langle f_1, f_2 \rangle_{H^2} + \langle g_1, g_2 \rangle_{H^2}.
\]
Then, for any $u, v \in H^2$,
\[
\begin{aligned}
\left\langle 
\mathrm{col}(u, Cu), 
\mathrm{col}(-C^* v, v)
\right\rangle
&= \langle u, -C^* v \rangle_{H^2} + \langle Cu, v \rangle_{H^2} \\
&= -\langle Cu, v \rangle_{H^2} + \langle Cu, v \rangle_{H^2} \\
&= 0.
\end{aligned}
\]
It follows that $\mathcal G(C)^\perp
= \left\{ \mathrm{col}(-C^* v, v) : v \in H^2 \right\}$~\cite[Prop. 1.2.5]{Feintuch1998}.
\end{proof}

\begin{lemma}[Gap and subspace intersection]
\label{app:thm:gapintersection}
Let $M,N\subseteq\mathcal H$ be finite-dimensional subspaces with
$\dim(M)=\dim(N)$. Then
\[
\mathrm{gap}(M,N)<1
\,\,\Longleftrightarrow\,\,
M^\perp\cap N=\{0\}
\,\,\Longleftrightarrow\,\,
M\cap N^\perp=\{0\},
\]
where $\mathrm{gap}(M,N)\eqdef\max\{\|P_{N^\perp}P_M\|,      \|P_{M^\perp}P_N\|\}$.
\end{lemma}

\begin{proof}
Since orthogonal projectors are contractions,
$\|P_{M^\perp}P_N\|\leq1$. In finite dimensions, the norm is attained, and $\|P_{M^\perp}P_N\|=1$ iff $M^\perp\cap N\neq\{0\}$.
Indeed, if a unit vector $x\in M^\perp\cap N$ exists, then
$P_{M^\perp}P_Nx=x$. Conversely, if the norm is one, there exists a
unit vector $x\in N$ such that $\|P_{M^\perp}x\|=1$, which implies
$x\in M^\perp$.

Similarly, $\|P_{N^\perp}P_M\|=1$ iff $M\cap N^\perp\neq\{0\}$. Finally, since $\dim(M)=\dim(N)$, the restricted projection
$P_{M|N}:N\rightarrow M$ is injective iff it is bijective.
Its kernel is $N\cap M^\perp$, while the kernel of its adjoint
$P_{N|M}:M\rightarrow N$ is $M\cap N^\perp$. Hence these two
intersections are trivial simultaneously. The result follows from the
definition of the gap.
\end{proof}

\begin{lemma}[Restricted Projection]\label{app:thm:restrictedprojector}
Let $M$ and $N$ be finite-dimensional subspaces with $\dim(N)\!=\!\dim(M)\!=\! n$, and assume $\delta\!=\! \mathrm{gap}(M,N)\!<\!1$. Let $P_{M|N}: N\!\mapsto\! M$ be the restricted projector that projects $n\!\in\! N$ onto $M$. Then, $P_{M|N}$ is bijective, $\|P_{M|N} n\|\geq \sqrt{1-\delta^2} \|n\|\,\forall n\!\in\! N$, and $\|P_{M|N}^{-1}\|\leq \frac{1}{\sqrt{1-\delta^2}}$.
\end{lemma}
\begin{proof}
To show injectivity, note that from Thm.~\ref{thm:gaptheorem}, $\mathrm{gap}(M,N)\!<\!1 \Leftrightarrow N\cap M^\perp = N^\perp\cap M = \lbrace 0 \rbrace$. Take $x\in N$ and suppose that $P_{M|N}x=P_M x = 0$. This implies that $x\in M^\perp$ and therefore $x\in M^\perp\cap N = \lbrace 0 \rbrace$. So $P_{M|N}x=0\Leftrightarrow x=0$. Bijectivity follows from $P_{M|N}$ being injective and linear. 

To obtain the inequalites, use 
\begin{align*}
\|P_{M|N} n\| = \|P_M n\| = \|n - P_{M^\perp} n\| = \|n - P_{M^\perp} P_N n\|,
\end{align*}
and lower bound using $\delta =\|P_{M^\perp} P_N\|$.
\end{proof}

\begin{remark}
A finite-horizon data representation of the subspace formed by $\mathcal G(C)^\perp$ can be constructed by truncating signals to length $L$. Given a signal $v$, define its sliding windows
\[
v_{[k,k+L-1]} := \begin{bmatrix} v_k & v_{k+1} & \cdots & v_{k+L-1} \end{bmatrix}.
\]
Then the corresponding Hankel-type matrix for $\mathcal G(C)^\perp$ is
\[
H_{L, C^\perp}
=
\begin{bmatrix}
-\,C^* v_{[0,L-1]} & -\,C^* v_{[1,L]} & \cdots & -\,C^* v_{[T-L,T-1]} \\
\;\;v_{[0,L-1]}    & \;\;v_{[1,L]}    & \cdots & \;\;v_{[T-L,T-1]}
\end{bmatrix}.
\]
\end{remark}

\end{document}